\documentclass[12pt,reqno]{amsart}
\usepackage[final]{graphicx}
\usepackage{amsfonts}
\usepackage{amsmath}
\usepackage{amssymb}
\usepackage{amsthm}

\newtheorem{theorem}{Theorem}[section]
\newtheorem{lemma}[theorem]{Lemma}

\newtheorem{proposition}[theorem]{Proposition}

\newcommand{\upchi}{\raise1pt\hbox{$\chi$}}
\newcommand{\R}{{\mathord{\mathbb R}}}

\newcommand{\Z}{{\mathord{\mathbb Z}}}
\newcommand{\N}{{\mathord{\mathbb N}}}
\newcommand{\E}{{\mathord{\mathbb E}}}
\newcommand{\PP}{{\mathord{\mathbb P}}}

\begin{document}

\title[The random displacement model]{Low energy properties\\ of the random displacement model}
\author[Baker]{Jeff Baker$^1$}
\author[Loss]{Michael Loss$^2$}
\author[Stolz]{G\"unter Stolz$^3$}

\address{$^1$ Southern Company Generation, 600 North 18th Street, Birmingham, Alabama 35291-2641, jeffbake@southernco.com}

\address{$^2$ Georgia Institute of Technology, School of Mathematics,
Atlanta, Georgia 30332-0160, loss@math.gatech.edu}

\address{$^3$ University of Alabama at Birmingham, Department of Mathematics, Birmingham, Alabama 35294-1170, stolz@uab.edu}

\date{\today}
\maketitle

\vspace{.3truein}
\centerline{\bf Abstract}

\medskip
{\sl We study low-energy properties of the random displacement model, a random Schr\"odinger operator describing an electron in a randomly deformed lattice. All periodic displacement configurations which minimize the bottom of the spectrum are characterized. While this configuration is essentially unique for dimension greater than one, there are infinitely many different minimizing configurations in the one-dimensional case. The latter leads to unusual low energy asymptotics for the integrated density of states of the one-dimensional random displacement model. For symmetric Bernoulli-distributed displacements it has a $1/\log^2$-singularity at the bottom of the spectrum. In particular, it is not H\"older-continuous. }

\section{Introduction}

We consider the so-called random displacement model, i.e.\ the random Schr\"odinger operator
\begin{equation} \label{eq:hamiltonian}
H_{\omega} = -\Delta + V_{\omega},
\end{equation}
where the random potential $V_{\omega}$ is given by displacing a single site potential $q$ from the points of $\Z^d$,
\begin{equation} \label{eq:potential}
V_{\omega}(x) = \sum_{i\in \Z^d} q(x-i-\omega_i).
\end{equation}
For the real-valued single site potential $q$ we assume $q\in L^{\infty}(\R^d)$ and supp$\,q \subset [-r,r]^d$ for some $r<1/2$. We also assume that $q$ is reflection symmetric at each coordinate hyperplane, i.e.\ symmetric in each variable with the remaining variables fixed.
Throughout this paper we will consider displacement configurations $\omega = (\omega_i)_{i\in \Z^d}$ such that $\omega_i \in [-d_{max},d_{max}]^d$ for all $i$, where $r+d_{max}=1/2$. The latter ensures that the displaced single site potentials in (\ref{eq:potential}) do not overlap.

While the random displacement model is a physically quite natural way to describe structural disorder in a solid, it is mathematically much less well understood than Anderson-type models, where the disorder enters in the form of random couplings at the single-site potentials. This is mostly due to the fact that the random displacement model depends non-monotonously on the random parameters (in quadratic form sense). Anderson-type models, on the other hand, are monotonous in this sense, at least if the single-site potentials have fixed sign. The consequential challenge in determining the spectral properties of the random displacement model, in particular the low energy behavior, lies in having to gain a deeper understanding of the interaction between multiple random parameters as well as the interplay between kinetic and potential energy. Spectral averaging arguments using individual random parameters, a common tool in the theory of Anderson models, are not available here.

This is one of the main reasons why it is not yet known if the multi-dimensional random displacement model is localized at the bottom of the spectrum (in $d=1$ localization at all energies follows from the results in \cite{DSS}, see also \cite{India}). The only known result on localization for the multi-dimensional random displacement model is due to Klopp \cite{Klopp}, who considered the semi-classical version $-h^2\Delta+V_{\omega}$ of the random displacement model and identified a localized region near the bottom of the spectrum for sufficiently small values of the semi-classical parameter $h$.

An attempt to understand the low energy properties of the random displacement model has to start with describing the mechanism which characterizes the bottom of the spectrum. Under the above assumptions this has been achieved in \cite{BLS1} by identifying the periodic displacement configuration $\omega^{min}$, see (\ref{eq:minimizer}) below, which leads to the minimum of the almost sure spectrum of $H_{\omega}$.

Here we continue our study of the low energy properties of the random displacement model by first characterizing the set of {\it all} minimizing periodic configurations. It turns out that $\omega^{min}$, up to trivial translations, is the unique minimizer in $d\ge 2$, while in $d=1$ there are infinitely many periodic minimizers, which can be explicitly characterized. These results are stated in Section~\ref{sec2} and proven in Section~\ref{sec3}.

We then move to studying the low energy asymptotics of the integrated density of states (IDS) for the random displacement model. Showing smallness of the IDS near the bottom of the spectrum, typically in the form of {\it Lifshits tails}, is an important step in all approaches to low-energy localization for multi-dimensional random Schr\"odinger operators. It is interpreted as showing that the bottom of the spectrum is a {\it fluctuation boundary}.

However, as we will show here, for the one-dimensional random displacement model the behavior of the IDS can be very different. If the displacements only take values $d_{max}$ or $-d_{max}$, both with equal probability, then the IDS $N(E)$ has a very strong singularity at the bottom of the spectrum,
\begin{equation} \label{eq:logbound1}
 N(E) \ge \frac{C}{\log^2(E-E_0)}
 \end{equation}
for $E\in (E_0,E_0+\varepsilon)$ and constants $C>0$ and $\varepsilon>0$, see Theorem~\ref{thm:bernoulli}. Here $E_0$ denotes the almost sure minimum of the spectrum of $H_{\omega}$. Thus the IDS is not even H\"older-continuous at $E_0$, a new phenomenon which to our knowledge has not been found for any other models of random operators. This and related results for the IDS of the one-dimensional random displacement model with other distributions of the displacements are proven in Section~\ref{sec4}. For example, the extreme behavior (\ref{eq:logbound1}) only appears for the given case of a symmetric Bernoulli distribution of the displacements, see Theorem~\ref{thm:nonbernoulli}. But, in $d=1$ and as long as the distribution of the displacements $\omega_n$ is symmetric, one never gets Lifshits tails (Theorem~\ref{thm:lifshitsexp}).

In Section~\ref{sec5} we comment on several open problems, in particular on our expectation that the uniqueness of the periodic minimizing configuration in $d\ge 2$ should indicate different low energy asymptotics for the IDS than in $d=1$ (e.g.\ the appearance of some form of Lifshits tails). We will also discuss related recent works by Klopp and Nakamura \cite{Klopp/Nakamura} and Fukushima \cite{Fukushima}.

\section{Periodic configurations which minimize the ground state energy} \label{sec2}

In \cite{BLS1} we have identified a simple periodic configuration of displacements which leads to the lowest possible spectral minimum for $H_{\omega}$ among all configurations $\omega$:

\begin{proposition}[Theorem~1.1 in \cite{BLS1}] \label{prop:BLS1a}
Let $\omega^{min}$ be given by
\begin{equation} \label{eq:minimizer}
\omega_i^{min} = ((-1)^{i_1} d_{max}, \ldots, (-1)^{i_d} d_{max})
\end{equation}
for all $i=(i_1,\ldots,i_d)\in \Z^d$. Then
\begin{equation} \label{eq:minproperty}
E_0 = \min \sigma(H_{\omega^{min}}),
\end{equation}
where $E_0 := \inf_{\omega} \min \sigma (H_{\omega})$.
\end{proposition}

This configuration is $2$-periodic in each coordinate, where in each period cell $2^d$ single sites cluster together in adjacent corners of unit cubes, see Figure~\ref{fig:minimizer} for $d=2$.

\begin{figure}[h]
  \centering
  \includegraphics[width=0.35\textwidth]{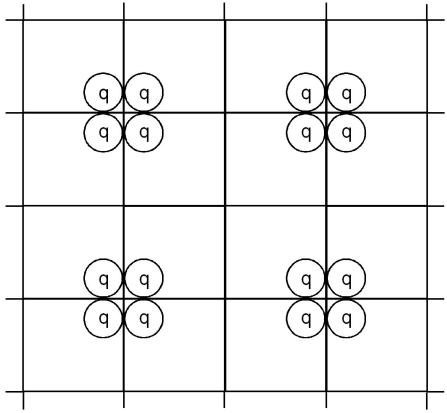}
  \caption{The support of $V_{\omega^{min}}$ for $d=2$.} \label{fig:minimizer}
\end{figure}

Here our first goal is to characterize {\it all} periodic configurations $\omega$ such that $\min \sigma(H_{\omega}) = E_0$. For this we will also use another result found in \cite{BLS1}: Let $q$ be as above, $a\in [-d_{max},d_{max}]^d$, $\Lambda_0 = (-\frac{1}{2}, \frac{1}{2})^d$ the unit cube centered at $0$, and $H_{\Lambda_0}^N(a) =-\Delta+q(x-a)$ on $L^2(\Lambda_0)$ with Neumann boundary condition on $\partial \Lambda_0$.

\begin{proposition}[Theorem~1.3 of \cite{BLS1}] \label{prop:BLS1b}
If
\begin{equation} \label{eq:mineval}
E_0(a) := \min \sigma(H_{\Lambda_0}^N(a)),
\end{equation}
then the following alternative holds: Either
\begin{itemize}
\item[(i)] $E_0(a)$ is strictly maximized at $a=0$ and strictly
minimized in the $2^d$ corners\\ $(\pm d_{max}, \ldots, \pm
d_{max})$ of $[-d_{max},d_{max}]^d$, or
\item[(ii)] $E_0(a)$ is identically zero. In this case the
corresponding eigenfunction is constant outside of the support of
$q$.
\end{itemize}
\end{proposition}

Alternative (i) holds for generic potentials, for example if $q$ is non-zero and sign-definite, in which case $E_0(a)$ never vanishes. But sign-definiteness of $q$ is far from necessary for alternative (i). Alternative (ii) holds if and only if the Neumann problem for $-\Delta +q$ restricted to $(-r,r)^d$ has ground state energy $0$ (if supp$\,q$ is simply connected with sufficiently regular boundary one may equivalently use the Neumann problem on supp$\,q$ here). Examples of sign-indefinite $q$ with this property can easily be constructed.

\begin{theorem} \label{thm:thm1}
(a) If alternative (ii) holds, then $\min \sigma(H_{\omega}) =0$ for all configurations $\omega= (\omega_i)_{i\in \Z^d}$, with $\omega_i \in [-d_{max},d_{max}]^d$ for all $i$.

(b) If alternative (i) holds, $d\ge 2$ and $r<1/4$, then $\omega^{min}$ given by (\ref{eq:minimizer}) is, up to translations, the unique periodic configuration with $\min \sigma(H_{\omega^{min}}) = E_0$.
\end{theorem}

We don't believe that the extra condition $r<1/4$ (beyond $r<1/2$) is required in Theorem~\ref{thm:thm1}, but we need it in our proof.

It remains to settle the case of alternative (i) and $d=1$. In this case the periodic minimizer is highly non-unique, but the set of all periodic minimizers can be characterized by our next result. Here, for $L\in \N$ let $S_L$ denote the set of all $L$-periodic configurations $(\omega_i)_{i\in \Z}$ such that $\omega_i =-d_{max}$ or $\omega_i=d_{max}$ for all $i$. Furthermore, for $\omega \in S_L$ let $n^{\pm}(\omega)$ be the number of $i\in \{1,\ldots, L\}$ with $\omega_i = \pm d_{max}$.

\begin{theorem} \label{thm:thm2}
Let $d=1$ and $q$ such that alternative (i) holds. An $L$-periodic configuration $\omega$ satisfies $\min \sigma(H_{\omega})=E_0$ if and only if $L$ is even, $\omega \in S_L$, and $n^-(\omega) =n^+(\omega)$.
\end{theorem}

Thus, in each period interval of $V_{\omega}$, equally many of the single site potentials sit at the extreme right and the extreme left of their allowed range of positions. The {\it dimer configuration} $\omega^{min}$ is merely a special case of this situation.

The short explanation behind Theorems~\ref{thm:thm1} and \ref{thm:thm2} is that for $d=1$ the ground states of $H_{\Lambda_0}^N(a)$ for extremal positions of $a$ can be made to match after suitable re-scaling, while the richer geometry for $d\ge 2$ prevents this for all configurations other than $\omega^{min}$.

\section{Proof of Theorems~\ref{thm:thm1} and \ref{thm:thm2}} \label{sec3}

Part (a) of Theorem~\ref{thm:thm1} easily follows from known facts: Under alternative (ii) the ground state energy of the Neumann problem for $-\Delta+q$ restricted to $(-r,r)^d$ is $0$. Let $\varphi$ be the corresponding positive normalized eigenfunction. For any given configuration $\omega$ place a translate of $\varphi$ at $i+\omega_i$ for each $i\in \Z^d$ and extend by a constant to the exterior of the convex hull of supp$\,q(\cdot- i- \omega_i)$. This results in a globally bounded, positive weak solution $\tilde{\varphi}$ of $-\Delta \tilde{\varphi} + V_{\omega} \tilde{\varphi}=0$. By Shnol's Theorem we have $0\in \sigma(H_{\omega})$. As $\tilde{\varphi}$ is positive it also follows from Theorem~C.8.1 in \cite{Simon} that $0\le \min \sigma(H_{\omega})$. Thus $\min \sigma(H_{\omega}) =0$ for all configurations $\omega$.

For the remainder of this section we will assume that alternative (i) holds. If $\omega$ is a periodic configuration with period $L=(L_1, \ldots, L_d) \in \N^d$, i.e.\ such that $\omega_{i+(n_1 L_1, \ldots, n_d L_d)} = \omega_i$ for all $i\in\Z^d$ and and $(n_1,\ldots,n_d) \in \Z^d$, we choose $\Lambda = (1/2,\ldots, L_1+1/2) \times \ldots (1/2, \ldots, L_d+1/2) \subset \R^d$ as period cell of the potential $V_{\omega}$ given by (\ref{eq:potential}). By $H_{\omega,\Lambda}^P$ and $H_{\omega,\Lambda}^N$ we denote the restrictions of $-\Delta+V_{\omega}$ to $L^2(\Lambda)$ with periodic and Neumann boundary conditions, respectively, and denote their lowest eigenvalues by $E_0(H_{\omega,\Lambda}^P)$ and $E_0(H_{\omega,\Lambda}^N)$.

The proofs of Theorems~\ref{thm:thm1}(b) and Theorem~\ref{thm:thm2} will be based on the following result which holds for arbitrary dimension. It shows, in particular, that in minimizing periodic configurations all single site potentials necessarily must sit in the corners of unit cubes centered at the points of $\Z^d$.

\begin{lemma} \label{lem:lem1}
Let $\omega$ be a periodic configuration with $\min \sigma(H_{\omega}) =E_0$. Then, for all $i\in \Z^d$, $\omega_i \in \{(a_1,\ldots,a_d)\in \R^d: a_k \in \{-d_{max}, d_{max}\} \:\mbox{for all $k=1,\ldots,d$}\}$. Moreover, in this case $E_0(H_{\omega,\Lambda}^P) = E_0(H_{\omega,\Lambda}^N)$ and the ground state eigenfunction $\psi_{\omega}$ of $H_{\omega,\Lambda}^N$ satisfies Neumann boundary conditions on the boundary of each unit cube $\Lambda_i$ centered at $i\in \Lambda \cap \Z^d$.
\end{lemma}

\begin{proof}
By assumption and Floquet-Bloch theory $E_0 = \min \sigma(H_{\omega}) =E_0(H_{\omega,\Lambda}^P)$. Also, by the variational principle, $E_0(H_{\omega,\Lambda}^P) \ge E_0(H_{\omega,\Lambda}^N)$. The ground state $\psi_{\omega}$ minimizes the quadratic form of $H_{\omega,\Lambda}^N$, thus
\begin{eqnarray} \label{eq:Neumannbrack}
E_0(H_{\omega,\Lambda}^N) & = & \frac{\int_{\Lambda} |\nabla \psi_{\omega}|^2 + \int_{\Lambda} \sum_{i\in \Lambda \cap \Z^d} q(x-i-\omega_i) |\psi_{\omega}|^2}{\int_{\Lambda} |\psi_{\omega}|^2} \nonumber \\
& = & \sum_{i\in \Lambda \cap \Z^d} \frac{\int_{\Lambda_i} |\nabla \psi_{\omega}|^2 + \int_{\Lambda_i} q(x-i-\omega_i) |\psi_{\omega}|^2}{ \int_{\Lambda_i} |\psi_{\omega}|^2} \cdot \frac{\int_{\Lambda_i} |\psi_{\omega}|^2}{\int_{\Lambda} |\psi_{\omega}|^2} \nonumber \\
& \ge & \sum_{i\in \Lambda \cap \Z^d} E_0(\omega_i) \frac{\int_{\Lambda_i} |\psi_{\omega}|^2}{\int_{\Lambda} |\psi_{\omega}|^2} \ge \sum_{i\in \Lambda \cap \Z^d} E_0 \frac{\int_{\Lambda_i} |\psi_{\omega}|^2}{\int_{\Lambda} |\psi_{\omega}|^2} = E_0,
\end{eqnarray}
where $E_0(\omega_i)$ is given by (\ref{eq:mineval}). In the second to last inequality we have used the variational principle as well as the fact that $\psi_{\omega}$ does not vanish on any of the $\Lambda_i$.

We conclude that all inequalities above must indeed be equalities, which immediately gives $E_0(H_{\omega,\Lambda}^P) = E_0(H_{\omega,\Lambda}^N)$. If for at least one $\omega_i \not\in \{a:a_k\in \{-d_{max}, d_{max}\}, k=1,\ldots,d\}$, then the last inequality in (\ref{eq:Neumannbrack}) would be strict, given that we are in alternative (i). We conclude that all $\omega_i$ sit in a corner. Finally, we see that
\begin{equation} \label{eq:cubeminimizer}
\frac{\int_{\Lambda_i} |\nabla \psi_{\omega}|^2 + \int_{\Lambda_i} q(x-i-\omega_i) |\psi_{\omega}|^2}{ \int_{\Lambda_i} |\psi_{\omega}|^2} = E_0(\omega_i)
\end{equation}
for each $i$. Thus the restriction of $\psi_{\omega}$ to $\Lambda_i$ is the ground state for the Neumann problem of $-\Delta +V(x-i-\omega_i)$ on $\Lambda_i$ and thus satisfies Neumann boundary conditions on each $\Lambda_i$.
\end{proof}

We now consider $d=1$, still under alternative (i), where we use one more lemma to prepare for the proof of Theorem~\ref{thm:thm2}.

\begin{lemma} \label{lem:lem2}
Let $d=1$ and $H_{\Lambda_0}^N(a)$ be the restriction of $-d^2/dx^2+q(x-a)$ to $L^2(-\frac{1}{2},\frac{1}{2})$ with Neumann boundary conditions. Let $\psi$ be the positive normalized ground state of $H_{\Lambda_0}^N(d_{max})$. Then $\psi(\frac{1}{2}) \not= \psi(-\frac{1}{2})$.
\end{lemma}

\begin{proof}
Suppose that $\psi(\frac{1}{2}) = \psi(-\frac{1}{2})$, then $\psi$ coincides with the periodic ground state of $-d^2/dx^2+q(x-d_{max})$ on $L^2(-\frac{1}{2},\frac{1}{2})$. Due to the symmetry of $q$ and the uniqueness of the periodic ground state, it must therefore also satisfy $\psi'(-\frac{1}{2}+d_{max})=0$ (considering $(-\frac{1}{2},\frac{1}{2})$ as a $1$-torus, $-\frac{1}{2}+d_{max}$ lies opposite to $d_{max}$). As $q(x)=0$ for $x\in (-\frac{1}{2}, -\frac{1}{2}+d_{max})$, we have that $-\psi''= E_0(d_{max})\psi$ on $[-\frac{1}{2}, -\frac{1}{2}+d_{max}]$ and
\begin{equation} \label{eq:psiprime}
\psi'(-\frac{1}{2}) = \psi'(-\frac{1}{2}+d_{max})=0.
\end{equation}
If $E_0(d_{max})<0$, then $\psi''=-E_0(d_{max})\psi >0$ and thus $\psi$ is strictly convex on $[-1/2,-1/2+d_{max}]$, contradicting (\ref{eq:psiprime}). Similarly, $E_0(d_{max})>0$ would yield strict concavity of $\psi$, again contradicting (\ref{eq:psiprime}). Thus $E_0(d_{max})=0$ and it follows from (\ref{eq:psiprime}) that $\psi$ must be constant outside the support of $q(x-d_{max})$. This contradicts that we have assumed alternative (i).
\end{proof}

We can now complete the proof of Theorem~\ref{thm:thm2}: Let $\omega$ be an $L$-periodic configuration which satisfies $\min \sigma(H_{\omega})=E_0$. By Lemma~\ref{lem:lem1} we have $\omega_i \in \{\pm d_{max}\}$ for all $i=1,\ldots,L$ and also $E_0 = \min \sigma(H_{\omega,\Lambda}^P) = \min \sigma(H_{\omega,\Lambda}^N)$, where $\Lambda = (\frac{1}{2},L+\frac{1}{2})$.

Let $u_D$ and $u_N$ be the solutions of $-u''+Vu=E_0u$ which satisfy $u_D(\frac{1}{2})=0$, $u_D'(\frac{1}{2})=1$, $u_N(\frac{1}{2})=1$, $u_N'(\frac{1}{2})=0$. $E_0$ is a Neumann eigenvalue on $\Lambda$ and thus $u_N'(L+\frac{1}{2})=0$. The transfer matrix for $H_{\omega}$ at $E_0$ from $1/2$ to $L+1/2$ is given by
\begin{equation}
T= \left( \begin{array}{cc} u_N(L+\frac{1}{2}) & u_D(L+\frac{1}{2}) \\ u_N'(L+\frac{1}{2}) & u_D'(L+\frac{1}{2}) \end{array} \right).
\end{equation}
This implies
\begin{equation} \label{eq:transferdet}
1 = \:\mbox{det}\,T = u_N(L+1/2) u_D'(L+1/2).
\end{equation}
Moreover, as $E_0$ is also an eigenvalue for periodic boundary conditions,
\begin{equation} \label{eq:transfertrace}
2 = \:\mbox{tr}\,T = u_N(L+1/2) + u_D'(L+1/2).
\end{equation}
We conclude from (\ref{eq:transferdet}) and (\ref{eq:transfertrace}) that
\begin{equation} \label{eq:neuper}
u_N(L+1/2) = 1 = u_D'(1/2),
\end{equation}
meaning that $u_N$ is both the Neumann and periodic eigenfunction to $E_0$.

We can use Lemma~\ref{lem:lem2} to understand the detailed structure of $u_N$: Let $\psi$ be the normalized ground state of $H_{\Lambda_0}^N(d_{max})$ as given there. Then, by symmetry of $q$, the normalized ground state of $H_{\Lambda_0}^N(-d_{max})$ is given by $\tilde{\psi}(x)= \psi(-x)$. As $\omega_i \in \{\pm d_{max}\}$ for all $i=1,\ldots, L$, we can construct $u_N$ by concatenating suitably re-scaled versions of $\psi$ and $\tilde{\psi}$, respectively, on the intervals $[i-\frac{1}{2}, i+\frac{1}{2}]$. With the positive number $r= \psi(\frac{1}{2})/\psi(-\frac{1}{2}) \not= 1$ from Lemma~\ref{lem:lem2} we thus have
\begin{equation} \label{eq:hopper}
u_N(i+1/2)/u_N(i-1/2) = \left\{ \begin{array}{ll} r & \mbox{if $\omega_i = d_{max}$}, \\ \frac{1}{r} & \mbox{if $\omega_i = -d_{max}$}. \end{array} \right.
\end{equation}
The accumulative effect is that $u_N(L+\frac{1}{2}) = r^{n_+(\omega)-n_-(\omega)} u_N(\frac{1}{2})$. We conclude from (\ref{eq:neuper}) that $n_+(\omega)=n_-(\omega)$, in particular that $L$ is even, which completes the proof of Theorem~\ref{thm:thm2}.

\vspace{.2in}

We now start with preparations for the proof of Theorem~\ref{thm:thm1}(b).

\begin{lemma} \label{lem:lem3}
Consider a connected open region $D$ in $\R^d$, $d\ge 2$ and a hyperplane $P$ that divides this region into
two nonempty subregions. Denote by $\sigma$ the reflection about $P$ and assume that $D\cap \sigma(D)$ is connected. Let $E\in \R$ and, in $D$, let $u$ be a solution of the equation
\begin{equation} \label{schroequ}
-\Delta u = E u
\end{equation}
which satisfies the condition $ \frac{\partial u}{\partial n} = 0$ on  $P\cap D $.
Then $u$ can be extended to a symmetric function $w$ on $D \cup \sigma (D)$ which satisfies the equation $-\Delta u = E u$ in this region.
\end{lemma}

\begin{proof}
Pick a point $x_0 \in P\cap D$, which we may assume to be the origin. Pick
a ball $B \subset D$ centered at the origin and pick coordinates $x_1, \dots, x_n$ so that $x'=(x_1, \dots, x_{n-1})$ are coordinates in $P$ and $x_n$ is the coordinate normal to $P$.
Consider the function
$$
v(x',x_n) = u(x',x_n) - u(x',-x_n)
$$
which satisfies  (\ref{schroequ}) and vanishes
on  $B \cap P$ identically. Its first normal derivative satisfies
\begin{equation}\label{first}
\frac{\partial v}{\partial x_n} (x',0) = 0 \ .
\end{equation}
Since
\begin{equation}
0 = E v(x',0) = -\Delta v(x',0) = -\frac{\partial^2 v}{\partial x_n^2} (x',0)
\end{equation}
we also have that
\begin{equation} \label{second}
\frac{\partial^2 v }{\partial x_n^2}(x',0) = 0 \ .
\end{equation}
Further, since
\begin{eqnarray}
&0& = E \left(\frac{\partial}{\partial x_n}v\right)(x',0) = -\left(\Delta \frac{\partial}{\partial x_n}v\right)(x',0) \\
&=&-\left(\frac{\partial^3}{\partial x_n^3}v\right)(x',0)
-\left(\Delta' \frac{\partial}{\partial x_n}v \right)(x',0)
\end{eqnarray}
we obtain from (\ref{first}) that
\begin{equation}\label{third}
\frac{\partial^3 v}{\partial x_n^3}(x',0) = 0 \ .
\end{equation}
Continuing in this fashion we deduce that
\begin{equation}
\frac{\partial^k v}{\partial x_n^k}(x',0)=0 \ ,
\end{equation}
for $k=0,1,2, \dots$.
In particular, this implies that all derivatives of $v$ vanish
at the origin. Since $v$ solves (\ref{schroequ})  it is a real
analytic function and hence it vanishes in the ball $B$. Since $x_0 \in D\cap \sigma(D)$ and $D\cap \sigma(D)$ is connected we learn that $v$ vanishes
everywhere in $D \cap \sigma(D)$. Thus $u$ is symmetric
in $D \cap \sigma(D)$ with respect to reflection about the plane $P$.
Next we prolong $u$ to the complement of $D \cap \sigma(D)$
in $D \cup \sigma(D)$ by setting
\begin{equation}
w(x) = \left\{
\begin{array}{r@{\quad:\quad}l}
u(x)  & x \in D,  \\
u(\sigma(x)) & x \in   \sigma(D).
\end{array}\right.
\end{equation}
Note that this function is defined since
$u(x)$ and $u(\sigma(x))$ coincide on $D \cap \sigma(D)$. Moreover, $x \in D\cup\sigma(D)$
means that $x \in D$ or $x \in \sigma(D)$ or both. In any case, by assumption and the fact that the Laplace operator commutes with reflections,  $w(x)$ satisfies the equation (\ref{schroequ}) at this point which proves the lemma.
\end{proof}

Given the previous lemma, we can now complete the proof of Theorem~\ref{thm:thm1}(b). Suppose that $d\ge 2$ and $\omega$ is a periodic configuration with $\min \sigma(H_{\omega})=E_0$, but not a translate of $\omega^{min}$. Then there must be two adjacent unit cubes, say $\Lambda$ and $\Lambda'$, such that the potential on the union $R=\Lambda \cup \Lambda'$ of these cubes is not symmetric with respect to reflection about their common face. This common face defines a hyperplane $P$. By Lemma~\ref{lem:lem1}, the two potential sites of $V_{\omega}$ restricted to $R$ are supported in corners of $\Lambda$ and $\Lambda'$, respectively, in the sense that $[a_1-r,a_1+r] \times \ldots \times [a_d-r,a_d+r]$ sits in a corner of $\Lambda$, and similar for $a'$ and $\Lambda'$. Also, the positive ground state eigenfunction $\psi_{\omega}$ of $H_{\omega}$ satisfies Neumann conditions on $\partial R$ as well as on $P$. Let
\begin{equation} \label{eq:reduceddomain}
D = R \setminus ([a_1-r,a_1+r]\times \ldots \times [a_d-r,a_d+r] \cup [a_1'-r,a_1'+r]\times \ldots \times [a_d'-r, a_d'+r])
\end{equation}
and $u$ the restriction of $\psi_{\omega}$ to $D$. For these choices of $D$ and $u$ and $E=E_0$ we can apply Lemma~\ref{lem:lem3}. In particular, the assumption $r<1/4$ assures the required connectedness of $D$ and $D\cap \sigma(D)$. Non-symmetry of $a$ and $a'$ and again $r<1/4$ implies that $D\cup \sigma(D)$ is all of $R$. Thus, by Lemma~\ref{lem:lem3}, $u$ can be extended to a symmetric function $w$ on $R$ which satisfies $-\Delta w = E_0 w$ on $R$. Therefore it is the ground state for the Neumann problem of $-\Delta$ on $R$. This implies $E_0=0$ and that $w$ is constant, a contradiction to the assumption of alternative (i).

\section{Consequences for the integrated density of states} \label{sec4}

We now consider the {\it random} displacement model, i.e.\ the model (\ref{eq:hamiltonian}), (\ref{eq:potential}) for the case where $\omega = (\omega_i)_{i\in\Z^d}$ is an array of i.i.d.\ vector-valued random variables with common distribution $\mu$ supported on $[-d_{max},d_{max}]^d$. Here, as usual, we define
\[
\mbox{supp}\,\mu := \{a\in \R^d: \mu(\{x:|x-a|<\varepsilon\})>0\:\mbox{for all}\: \varepsilon>0\}.
\]
Then the random operator $H_{\omega}$ is ergodic with respect to translations in $\Z^d$ and thus has all the basic properties of ergodic operators, see e.g.\ \cite{Carmona/Lacroix}. In particular, the integrated density of states (IDS)
\begin{equation} \label{eq:ids}
N(E) = \lim_{L\to\infty} \frac{1}{|\Lambda_L|} \E(\mbox{tr}\, \chi_{(-\infty,E]}(H_{\omega,\Lambda_L}^X))
\end{equation}
exists for all energies $E\in\R$. Here $\Lambda_L = (\frac{1}{2},L+\frac{1}{2})^d$ and $H_{\omega,\Lambda_L}^X$ is the restriction of $H_{\omega}$ to $L^2(\Lambda_L)$ with boundary condition $X\in \{P,N,D\}$, as periodic (P), Neumann (N) and Dirichlet (D) boundary conditions all give the same limit in (\ref{eq:ids}).

The spectrum $\sigma(H_{\omega})$ is almost surely deterministic, i.e.\ $\Sigma = \sigma(H_{\omega})$ for almost every $\omega$, and given by the growth points of the non-decreasing function $N(E)$. It can be characterized in terms of the spectra of those $H_{\omega}$ for which the configuration $\omega$ is periodic,
\begin{equation} \label{eq:periodsupp}
\Sigma = \overline{\bigcup_{\omega} \sigma(H_{\omega})},
\end{equation}
where the union is taken over all periodic $\omega$ such that $\omega_i \in$ supp$\,\mu$ for all $i$.
This corresponds to a well-known result for the Anderson model, e.g.\ \cite{Carmona/Lacroix}, and is found with the same proof. If the support of the distribution $\mu$ contains all the corners $\{(\pm d_{max}, \ldots, \pm d_{max})\}$ of the cube $[-d_{max},d_{max}]^d$, then it follows from (\ref{eq:periodsupp}) and Proposition~\ref{prop:BLS1a} that
\[ \min \Sigma = E_0 = \min \sigma(H_{\omega^{min}}). \]

For large classes of random Schr\"odinger operators it is known that the IDS vanishes rapidly at the bottom of the spectrum $E_0$, for example one has Lifshits tail behavior
\begin{equation} \label{eq:lifshits}
N(E) \sim e^{-c|E-E_0|^{-d/2}}
\end{equation}
for Anderson models with sign-definite single-site potential, see e.g.\ \cite{Kirsch}, \cite{Pastur/Figotin}0-ppp or \cite{Stollmann} for proper statements, proofs and references to the original literature.

It turns out that for the random displacement model the behavior of the IDS at the bottom of the spectrum is much more subtle. Here we will present several results for the one-dimensional displacement model, which were obtained in \cite{Baker}. We will generally assume that alternative (i) holds.

In the first result we will consider the one-dimensional Bernoulli displacement model, i.e.\ the case where the distribution of the displacements $\omega_i$ is given by
\begin{equation} \label{eq:bernoulli}
\mu = \frac{1}{2}\delta_{d_{max}} + \frac{1}{2} \delta_{-d_{max}}.
\end{equation}
It turns out that in this case the low-energy asymptotics of the IDS is at the opposite extreme of Lifshits tails:

\begin{theorem} \label{thm:bernoulli}
Let $H_{\omega}$ be the one-dimensional {\it symmetric Bernoulli} displacement model given by (\ref{eq:hamiltonian}), (\ref{eq:potential}) and (\ref{eq:bernoulli}) and assume that alternative (i) holds. Then there exist $C>0$ and $\varepsilon>0$ such that
\begin{equation} \label{eq:logbound}
N(E) \ge \frac{C}{\log^2(E-E_0)}
\end{equation}
for $E\in (E_0,E_0+\varepsilon)$.
\end{theorem}

As $N(E_0)=0$, this means that $N(E)$ has infinite upper derivative at $E=E_0$, i.e.\ the density of states $n(E)=N'(E)$ has a strong singularity at the bottom of the spectrum. This is opposed to the case of Lifshits tails which would yield $n(E_0)=0$. In fact, (\ref{eq:logbound}) says that the IDS is not even H\"older continuous at $E=E_0$, an even stronger singularity than one gets for the Laplacian $H_0 = -d^2/dx^2$, where the IDS has a van Hove singularity $C|E|^{1/2}$. For general one-dimensional ergodic Schr\"odinger operators (and for discrete ergodic Schr\"odinger operators also in higher dimension) the IDS is log-H\"older-continuous at all energies, i.e.\
\begin{equation} \label{eq:logholder}
|N(E)-N(E')| \le \frac{C}{|\log |E-E'||}
\end{equation}
for $E$ close to $E'$, see \cite{Craig/Simon, Craig/Simon2}. Craig and Simon constructed examples of quasi-periodic potentials which show that the bound (\ref{eq:logholder}) is optimal. As far as we know, the result in Theorem~\ref{thm:bernoulli} provides the first known example of a random potential (with finite correlation length) where for at least one energy the IDS is not H\"older-continuous and, in fact, close to the minimal possible regularity for ergodic operators given by (\ref{eq:logholder}).

\begin{proof}
We will use the standard lower bound, e.g.\ \cite{Carmona/Lacroix},
\begin{equation} \label{eq:IDSlow}
N(E) \ge \frac{1}{L} \PP(E_0(H_{\omega,L}^D)<E),
\end{equation}
which holds for arbitrary $L$, to be chosen later depending on $E$. Here $H_{\omega,L}^D$ is short for $H_{\omega,\Lambda_L}^D$, $\Lambda_L=(1/2,L+1/2)$.

To show that $E_0(H_{\omega,L}^D)<E$ we will find $\psi_{\omega} \in D(H_{\omega,L}^D)$ with $\|\psi_{\omega}\|=1$ and $\langle \psi_{\omega}, H_{\omega,L}^D \psi_{\omega} \rangle <E$. To construct $\psi_{\omega}$, let displacements $\omega = (\omega_1, \ldots, \omega_L)$ be given and let $u_N$ be the solution of $-u''+V_{\omega}u=E_0 u$ with $u_N(\frac{1}{2})=1$, $u_N'(\frac{1}{2})=0$. Choose cut-off functions $\theta_L \in C_0^{\infty}(\R)$ with $0\le \theta_L \le 1$, supp$\,\theta_L \subset [1,L]$, $\theta_L(x)=1$ for $3/2\le x\le L-1/2$, and $\|\theta_L'\|_{\infty}$ and $\|\theta_L''\|_{\infty}$ uniformly bounded in $L$.

As the $\omega_i$ have distribution (\ref{eq:bernoulli}), we have $\omega_i\in \{\pm d_{max}\}$ for all $i$. Thus the restriction of $-d^2/dx^2+V_{\omega}$ to $(i-1/2,i+1/2)$ has Neumann ground state energy $E_0$ for all $i\in\{1,\ldots, L\}$, which implies that
\begin{equation} \label{eq:zeroderiv}
u_N'(i+1/2)=0 \quad \mbox{for all $i\in \{1,\ldots,L\}$}.
\end{equation}
We choose $\psi_{\omega} := \theta_L u_N/ \|\theta_L u_N\|$ and calculate
\begin{eqnarray} \label{eq:approxef}
\langle \psi_{\omega}, H_{\omega,L}^D \psi_{\omega} \rangle -E_0 & = & \frac{\langle \theta_L u_N, -\theta_L''u_N\rangle -2\langle \theta_L u_N, \theta_L' u_N' \rangle}{\|\theta_Lu_N\|^2} \nonumber \\
& \le & \frac{\tilde{\beta} (1+u_N^2(L+1/2))}{\int_{3/2}^{L-1/2} u_N^2(x)\,dx} \nonumber \\
& \le & \frac{\beta(1+u_N^2(L+1/2))}{\sum_{i=1}^L u_N^2(i+1/2)}
\end{eqnarray}
where $\tilde{\beta}>0$ and $\beta>0$ can be chosen uniformly in $\omega$ and $L$. Here we have repeatedly used standard a priori upper and lower bounds on solutions of $-u''+Vu=Eu$, for example that $r(x) \sim r(x+1)$ and $\int_x^{x+1} u^2 \sim r^2(x)$, where $r(x)=(u^2(x)+u'^2(x))^{1/2}$ is the Pr\"ufer amplitude of $u$ and constants can be chosen uniform as long as $E$ and $\|V\|_{\infty}$ vary in a bounded interval, see e.g.\ \cite{Stolz1,Stolz2} for more details. We also use that by (\ref{eq:zeroderiv}) the Pr\"ufer amplitude of $u_N$ at the points $i+1/2$ coincides with $u_N(i+1/2)$.

Thus
\begin{equation} \label{eq:probbound}
\PP(E_0(H_{\omega,L}^D) < E) \ge \PP\left(\beta \frac{1+u_N^2(L+1/2)}{\sum_{i=1}^L u_N^2(i+1/2)} < E-E_0\right).
\end{equation}
Another consequence of $\omega_i\in \{\pm d_{max}\}$ is that $u_N$ satisfies (\ref{eq:hopper}) for every $i$ with a positive $r\not= 1$, using that we are in alternative (i). Assume without restriction that $r>1$ (if $0<r<1$ then we can do the following construction from ``right to left'', choosing $u_N(L+1/2)=1$, $u_N'(L+1/2)=0$) and set
\begin{equation} \label{eq:symmBernoulli}
X_i = \frac{\log (u_N(i+1/2)/u_N(i-1/2))}{\log r},
\end{equation}
$i=1,\ldots,L$. The $X_i$ are independent symmetric Bernoulli random variables with values $\pm 1$, and
\begin{equation} \label{eq:randomwalk}
u_N^2(i+1/2) = e^{2S_i \log r},
\end{equation}
where $S_i = X_1+\ldots +X_i$. If $Y:= \max_{i=1,\ldots,L} S_i$, then it is a consequence of the reflection principle for symmetric random walks, e.g.\ \cite{Feller} that
\begin{equation} \label{eq:reflection}
\PP(Y\ge \sqrt{L}| S_L\le 0) = \PP(S_L \ge 2\sqrt{L}).
\end{equation}
The latter converges to $\pi^{-1/2} \int_2^{\infty} \exp(-y^2/2)\,dy>0$ as $L\to\infty$ by the central limit theorem.

Let $A_L:= \{\omega| Y\ge \sqrt{L} \:\mbox{and $S_L\le 0$}\}$. If $Y\ge \sqrt{L}$, then $\sum_{i=1}^L u_N^2(L+1/2) \ge \exp(2\sqrt{L} \log r)$. Also, $S_L\ge 0$ means $u_N^2(L+1/2) \le 1$. Thus (\ref{eq:probbound}) implies
\begin{eqnarray} \label{eq:condprob}
\PP(E_0(H_{\omega,L}^D)<E) & \ge & \PP\left(\beta \frac{1+u_N^2(L+1/2)}{\sum_{i=1}^L u_N^2(i+1/2)} < E-E_0| A_L\right) \PP(A_L) \nonumber \\
& = & \PP(A_L) \ge c_0 >0
\end{eqnarray}
if $2\beta \exp(-2\sqrt{L}\log r) <E-E_0$ and $L$ sufficiently large. This determines the choice of $L\in \N$ for given $E$ such that
\begin{equation} \label{eq:Lchoice}
\frac{1}{2\beta} e^{-2\sqrt{L-1} \log r} \ge E-E_0 \ge \frac{1}{2\beta} e^{-2\sqrt{L} \log r}.
\end{equation}
Thus $L \sim \left( \frac{\log 2\beta (E-E_0)}{\log r}\right)^2$. From (\ref{eq:IDSlow}) and (\ref{eq:condprob}) we have $N(E) \ge c_0/L$, which, for $E-E_0$ sufficiently small, takes the form (\ref{eq:logbound}).

\end{proof}

As mentioned above, (\ref{eq:logbound}) says in particular that the IDS is not H\"older continuous at $E=E_0$. This is only possible if the distribution $\mu$ is concentrated in the extreme points $d_{max}$ and $-d_{max}$, as is demonstrated by our next result.

\begin{theorem} \label{thm:nonbernoulli}
Suppose that the distribution $\mu$ of the $\omega_i$ in the one-dimensional displacement model (\ref{eq:hamiltonian}), (\ref{eq:potential}) satisfies
\begin{equation} \label{eq:nonbernoulli}
\mu((-d_{max},d_{max}))>0.
\end{equation}
Then the IDS $N(E)$ is H\"older continuous at $E=E_0$.
\end{theorem}

This result may not be optimal. We expect that under the conditions of Theorem~\ref{thm:nonbernoulli} one can at show that $N(E) \le C_{\alpha} |E-E_0|^{\alpha}$ near $E_0$ for arbitrary $\alpha>0$. But, as long as the distribution $\mu$ is chosen symmetric and not too small at $\pm d_{max}$, one does not get Lifshits tail decay as in (\ref{eq:lifshits}). To make this precise, define the {\it Lifshits exponent} $\gamma$ at $E_0$ by
\begin{equation} \label{eq:lifshitsexp}
\gamma = \lim_{E\downarrow E_0} \frac{\log(-\log N(E))}{\log(E-E_0)}
\end{equation}
whenever this limit exists. Note that $\gamma\le 0$. If $\gamma<0$, then it determines the asymptotics of the IDS in the sense that, up to logarithmic corrections, $N(E) \sim C_1 \exp(-C_2(E-E_0)^{\gamma})$ as $E\downarrow E_0$.

\begin{theorem} \label{thm:lifshitsexp}
Assume that the distribution $\mu$ is symmetric and satisfies
\begin{equation} \label{eq:mutails}
\mu([d_{max}-\epsilon,d_{max}] \cup [-d_{max}, -d_{max}+\epsilon]) \ge C_1\epsilon^N
\end{equation}
for some positive $C_1$ and $N$ and all $\epsilon>0$. Also assume that the single-site potential $q$ is uniformly h\"older continuous, i.e.\ that $|q(x)-q(y)| \le C_2 |x-y|^{\rho}$ for some $C_2$ and $\rho>0$ and all $x$, $y$.

Then $\gamma=0$.
\end{theorem}

In the following we sketch the proofs of Theorems~\ref{thm:nonbernoulli} and \ref{thm:lifshitsexp}, referring for additional details to \cite{Baker}.

\vspace{.3cm}

To \emph{prove Theorem~\ref{thm:lifshitsexp}} we follow the general strategy of the proof of Theorem~\ref{thm:bernoulli}, starting with (\ref{eq:IDSlow}) and using the test function $\psi_{\omega} := \theta_L u_N/\|\theta_L u_N\|$. However, the construction of $u_N$ needs to be modified as follows: On each interval $[i-1/2,i+1/2]$, $i\in\{1,\ldots,L\}$, $u_N$ is chosen to coincide with a constant multiple of the positive ground state of the Neumann problem for $-d^2/dx^2+q(x-i-\omega_i)$ on $[i-1/2,i+1/2]$. Scaling constants are chosen such that $u_N(1/2)=1$ and $u_N$ is continuously differentiable throughout $[1/2,L+1/2]$. As we now generally have $E_0(\omega_i) \not= E_0$, this leads to extra terms in the bound
\begin{eqnarray} \label{eq:IDSmodlower}
N(E) & \ge & \frac{1}{L} \PP\left( \frac{|\langle \theta_L \psi_{\omega}, -\theta_L''\psi_{\omega}\rangle| + 2|\langle \theta_L \psi_{\omega}, \theta_L' \psi_{\omega}'\rangle |}{\|\theta_L \psi_{\omega}\|^2} \right.\nonumber \\
& & \left.\mbox{}+ \sum_{i=1}^L (E_0(\omega_i)-E_0) \frac{\|\theta_L \psi_{\omega}\|^2}{\|\theta_L \psi_{\omega}\|^2} < E-E_0 \right),
\end{eqnarray}
here $\|\theta_L \psi_{\omega}\|_1^2 := \int_{i-1}^i \theta_L^2 \psi_{\omega}^2$.

Due to the symmetry of $\mu$, the numbers $\log u_N(i+1/2)$, $i=1,\ldots,L$, are still a symmetric random walk (but not Bernoulli). Versions of the reflection principle and central limit theorem for general symmetric random walks and a choice of $L$ as in (\ref{eq:Lchoice}) (with a suitable positive constant replacing $\log r$) lead to the bound
\begin{eqnarray} \label{eq:IDSmodlower2}
N(E) & \ge & \frac{C}{L} \PP \left( \sum_{i=1}^L (E_0(\omega_i)-E_0) \frac{\|\theta_L \psi_{\omega}\|_i^2}{\|\theta_L \psi_{\omega}\|^2} < E-E_0 \right) \nonumber \\
& \ge & \frac{C}{L} \PP \left( \sum_{i=1}^L (E_0(\omega_i)-E_0) < E-E_0 \right) \nonumber \\
& \ge & \frac{C}{L} \left( \PP (E_0(\omega_1)-E_0 < \frac{E-E_0}{L}) \right)^L.
\end{eqnarray}
In Lemma~2.1 of \cite{BLS1} the continuity of $E_0(\cdot)$ was shown. The proof given there provides the bound
\[
|E_0(a_1)-E_0(a_2)|^p \le C \int |q(x-a_1)-q(x-a_2)|^p\,dx
\]
for any $p\ge 2$. Uniform h\"older continuity of $q$ gives h\"older continuity of $E_0(\cdot)$. Using $E_0 = E_0(d_{max}) = E_0(-d_{max})$ and (\ref{eq:mutails}) we see that $\PP(E_0(\omega_1)-E_0 <\delta) \ge C_1(\delta/C)^{N/\rho}$. We plug into (\ref{eq:IDSmodlower2})
\[
N(E) \ge \frac{C}{L} \left( \frac{E-E_0}{L} \right)^{N/\rho}.
\]
From this bound, having chosen $L$ through (\ref{eq:Lchoice}), a calculation shows that the Lifshits exponent vanishes.

\vspace{.3cm}

The proof of Theorem~\ref{thm:nonbernoulli} is based on the standard upper bound, e.g.\ \cite{Carmona/Lacroix},
\begin{equation} \label{eq:standardup}
N(E) \le C \PP (E_0(H_{\omega,L}^N) \le E).
\end{equation}
We choose $L$ through
\begin{equation} \label{eq:Lchoice2}
s_0 e^{-2C_1(L+1)} \le E-E_0 \le s_0 e^{-2C_1L}
\end{equation}
with constants $s_0$ and $C_1$ to be determined later. By the calculation done in (\ref{eq:Neumannbrack}),
\begin{equation}
E_0(H_{\omega,L}^N) \ge \sum_{i=1}^L E_0(\omega_i) \frac{\int_{i-1}^i |\psi_{\omega}|^2}{\int_{1/2}^{L+1/2} |\psi_{\omega}|^2},
\end{equation}
where $\psi_{\omega}$ is the ground state of $H_{\omega,L}^N$. By a priori bounds (e.g.\ \cite{Stolz1}) there exists $C_1>0$ such that
\[
e^{-C_1L} \le \int_{i-1}^i |\psi_{\omega}|^2\,dx \le e^{C_1L}
\]
uniformly in $L\in \N$, $i\in \{1,\ldots,L\}$ and all configurations $\omega$. Using this $C_1$ in (\ref{eq:Lchoice2}) we further estimate
\[
\PP(E_0(H_{\omega,L}^N) \le E) \le \PP\left( \sum_{i=1}^L \frac{E_0(\omega_i)-E_0}{L} \le s_0\right) \le e^{-\gamma_0 L}.
\]
Here the last step is a large deviations bound, which is applicable with suitably chosen $s_0>0$ and $\gamma_0>0$ due to the assumption (\ref{eq:nonbernoulli}). Note for this that $E(\omega_i)-E_0$ are non-negative random variables which are strictly positive with positive probability. With this $s_0$ in (\ref{eq:Lchoice2}) if follows that $e^{-\gamma_0L} \le C(E-E_0)^{\alpha}$, where $\alpha:=s_0/4C_1$. This completes the proof.

\section{Concluding remarks} \label{sec5}

With the above results we have only started to touch the various possibilities for the low-energy asymptotics of the IDS in the random displacement model. There are several other regimes which we haven't considered yet:

(i) For one-dimensional random displacement models with {\it non-symmetric} distribution, in particular the case $\mu = p \delta_{d_{max}} + (1-p) \delta_{-d_{max}}$ with $p\not= 1/2$ we expect that the IDS might have Lifshits tails.

(ii) It would be most interesting to decide if the uniqueness of the minimizing periodic configuration established in Theorem~\ref{thm:thm2}(b) leads to Lifshits tails of the IDS at $E_0$ for the multi-dimensional random displacement model with general (or suitable) distributions $\mu$. Beyond uniqueness of the minimizing configuration this would require to have quantitative results on the probability that other configurations have ground state energy near $E_0$.

In this context we mention the recent work of Klopp and Nakamura \cite{Klopp/Nakamura} on sign-indefinite Anderson models, where some phenomena similar to those found by us for the random displacement model appear. In particular, they find that Lifshits tail as well as van Hove asymptotics of the IDS at the bottom of the spectrum are both possible in their model, depending on the choice of single-site potential and distribution of the random parameters. They have informed us about work in preparation \cite{Klopp/Nakamura2} which, when combined with the uniqueness result Theorem~\ref{thm:thm2}(b) above, should indeed lead to Lifshits-type asymptotics of the IDS for multi-dimensional random displacement models as considered here. This will need suitable assumptions on the distribution $\mu$ of the displacements, namely that $\mu$ is concentrated on the corners of $[-d_{max},d_{max}]^d$.

(iii) We have used the non-overlap condition supp$\,\mu \subset [-d_{max},d_{max}]^d$, $d_{max}+r=1/2$, mostly for technical reasons. In particular, it is crucial for the Neumann-bracketing arguments used in \cite{BLS1} and also Section~\ref{sec3} above. However, relaxing this condition will also lead to different phenomena. We mention the recent work by Fukushima \cite{Fukushima} who studies the random displacement model (\ref{eq:hamiltonian}), (\ref{eq:potential}) for positive $q$ and displacements with unbounded distribution $\mu$. In this case it is easily seen that the almost sure spectrum is $[0,\infty)$, due to the presence of large empty regions in typical single-site configurations (while in our setting the spectral minimum would be strictly positive). Under this condition Fukushima establishes Lifshits tails of the IDS at $0$. Another interesting task would be to look at intermediate cases, where supp$\,\mu$ is bounded but not small, allowing overlapping finite clusters of single-site potentials, but no large empty regions.

(iv) Under alternative (ii) all random configurations give the same ground state energy. This is an example of a random operator with a stable spectral boundary (as opposed to fluctuation boundaries). In other examples of this type, for a discussion see Sections 6B and 9 of \cite{Pastur/Figotin}, this has been found to lead to van Hove behavior of the IDS, i.e.\ $N(E) \sim (E-E_0)^{d/2}$ as for the unperturbed Laplacian. We also expect this here.

\vspace{.5cm}

\noindent {\bf Acknowledgements:} M.\ L.\
would like to acknowledge partial support through NSF grant DMS-0600037. G.\ S.\ was partially supported through NSF grand DMS-0653374.


\bigskip
\begin{thebibliography}{A}

\bibitem{Baker} J.\ Baker, \emph{Spectral Properties of Displacement Models}, PhD thesis, University of Alabama at Birmingham, 2007, electronically available at www.mhsl.uab.edu/dt/2007p/baker.pdf

\bibitem{BLS1} J.\ Baker, M.\ Loss and G.\ Stolz, \emph{Minimizing the ground state energy in a randomly deformed lattice}, arXiv:0707.3988, to appear in Comm.~Math.~Phys.\

\bibitem{Carmona/Lacroix} R.\ Carmona and J.\ Lacroix, \emph{Spectral Theory of Random Schr\"odinger Operators}, Birkh\"auser, Basel, 1990

\bibitem{Craig/Simon} W.\ Craig and B.\ Simon, \emph{Subharmonicity of the Lyapunov exponent}, Duke~Math.~J.~{\bf 50} (1983), 551-560

\bibitem{Craig/Simon2} W.\ Craig and B.\ Simon, \emph{Log H\"older continuity of the integrated density of states for stochastic Jacobi matrices}, Comm.~Math.~Phys.~{\bf 90} (1983), 207-218

\bibitem{DSS} D.\ Damanik, R.\ Sims and G.\ Stolz, \emph{Localization for one-dimensional, continuum, Bernoulli-Anderson models}, Duke~Math.~J.~{\bf 114} (2002), 59-100

\bibitem{Feller} W. Feller, \emph{An Introduction to Probability Theory and Its Applications, Volume I}, John Wiley \&
Sons, Inc., 1968

\bibitem{Fukushima} R.\ Fukushima \emph{Brownian survival and Lifshitz tail in perturbed lattice disorder}, arXiv:0807.2486

\bibitem{Kirsch} W.\ Kirsch, \emph{An invitation to random Schr\"odinger operators}, arXiv:0709.3707

\bibitem{Klopp} F.\ Klopp, \emph{Localization for semiclassical continuous random Schr\"odinger operators. II. The random displacement model}, Helv.~Phys.~Acta~{\bf 66} (1993), 810-841

\bibitem{Klopp/Nakamura} F.\ Klopp and S.\ Nakamura, \emph{Spectral extrema and Lifshitz tails for non monotonous alloy type models}, arXiv:0804.4079

\bibitem{Klopp/Nakamura2} F.\ Klopp and S.\ Nakamura, private communication

\bibitem{Pastur/Figotin} L.\ Pastur and A.\ Figotin, \emph{Spectra of Random and Almost-Periodic Operators}, Springer, Berlin, 1992

\bibitem{Simon} B.\ Simon, \emph{Schr\"odinger semigroups}, Bull.~Amer.~Math.~Soc.~{\bf 7} (1982), 447-526

\bibitem{Stollmann} P.\ Stollmann, \emph{Caught by disorder --- Bound states in random media}, Vol.~20 of Progress in Math.~Physics, Birkh\"auser, Boston, 2001

\bibitem{Stolz1} G.\ Stolz, \emph{On the absolutely continuous spectrum of perturbed Sturm-Liouville operators}, J.~Reine~Angew.~Math.~{\bf 416} (1991), 1-23

\bibitem{Stolz2} G.\ Stolz, \emph{Bounded solutions and absolute continuity of Sturm-Liouville operators}, J.~Math.~Anal.~Appl.~{\bf 169} (1992), 210-228

\bibitem{India} G.\ Stolz, {\emph Strategies in localization proofs for one-dimensional random Schrödinger operators}, Spectral and inverse spectral theory (Goa, 2000).  Proc.~Indian~Acad.~Sci.~Math.~Sci.~{\bf 112} (2002), 229-243

\end{thebibliography}
\end{document}